\documentclass[conference]{IEEEtran}
\IEEEoverridecommandlockouts
\usepackage{cite}
\usepackage{amsmath,amssymb,amsfonts}
\usepackage{graphicx}
\usepackage{textcomp}
\usepackage{xcolor}

\usepackage{amsthm}
\theoremstyle{definition} 
\newtheorem{definition}{Definition}
\theoremstyle{plain} 
\newtheorem{theorem}{Theorem}[section] 

\usepackage{subfigure}
\usepackage{caption}

\usepackage{latexsym}

\usepackage[ruled,linesnumbered,vlined]{algorithm2e}
\usepackage[normalem]{ulem}
\usepackage{algpseudocode} 

\def\BibTeX{{\rm B\kern-.05em{\sc i\kern-.025em b}\kern-.08em
    T\kern-.1667em\lower.7ex\hbox{E}\kern-.125emX}}
\begin{document}

\title{Distributed Processing of $k$NN Queries over Moving Objects on Dynamic Road Networks}

\author{  
\IEEEauthorblockN{  
Mingjin Tao\IEEEauthorrefmark{1},  
Kailin Jiao\IEEEauthorrefmark{1},  
Yawen Li\IEEEauthorrefmark{2},  
Wei Liu\IEEEauthorrefmark{1},
Ziqiang Yu\IEEEauthorrefmark{1*}\thanks{*Corresponding author}}  

\IEEEauthorblockA{  
\IEEEauthorrefmark{1}Yantai University \IEEEauthorrefmark{2}Beijing University of Posts and Telecommunications}  
}
\DeclareRobustCommand*{\IEEEauthorrefmark}[1]{%
    \raisebox{0pt}[0pt][0pt]{\textsuperscript{\footnotesize\ensuremath{#1}}}}

\maketitle

\begin{abstract}
The $k$ Nearest Neighbor ($k$NN) query over moving objects on road networks is essential for location-based services. Recently, this problem has been studied under road networks with distance as the metric, overlooking fluctuating travel costs. We pioneer the study of the $k$NN problem within dynamic road networks that account for evolving travel costs. Recognizing the limitations of index-based methods, which become quickly outdated as travel costs change, our work abandons indexes in favor of incremental network expansion on each snapshot of a dynamic road network to search for $k$NNs. To enhance expansion efficiency, we present D$k$NN, a distributed algorithm that divides the road network into sub-networks for parallel exploration using Dijkstra's algorithm across relevant regions. This approach effectively addresses challenges related to maintaining global distance accuracy during local, independent subgraph exploration, while minimizing unnecessary searches in irrelevant sub-networks and facilitating the early detection of true $k$NNs, despite the lack of constant global search monitoring. Implemented on the Storm platform, D$k$NN demonstrates superior efficiency and effectiveness over traditional methods in real-world road network scenarios.
\end{abstract}


\section{Introduction}
Given a road network with a set of moving objects, a $k$-nearest neighbor ($k$NN) query seeks to identify the $k$ closest objects to a specified query point. This problem is fundamental to many location-based services~\cite{ref35,ref36,ref37,ref38}, such as ride-hailing, location-based gaming (e.g., Ingress), and emergency medical dispatch. In real-world scenarios, users are often more concerned with travel costs than road distances. Considering the travel cost of each road constantly varies as time evolves, we model the road network as a dynamic graph, where intersections and endpoints are represented as vertices, roads as edges, and the time-varying travel cost on each road as an evolving weight. Our focus is on processing $k$NN queries within this dynamic graph scenario, which more accurately reflects real-world conditions but has received limited attention in existing research.

The $k$NN problem over moving objects on the graph has been extensively studied~\cite{ref2,ref1,ref10,ref11,ref9,ref14,ref12,ref15,ref8,ref7,ref39,ref33,ref3,ref4,ref5,ref13,ref16,ref17,ref42,ref43}. Some of them adopt an Incremental Network Expansion (INE) strategy~\cite{ref2} that incrementally explores outward from the query point on the graph, gradually identifying the closest objects based on Dijkstra’s relaxation principle~\cite{ref1}. These methods typically terminate early once the $k$NNs are found. However, vertex-by-vertex expansion becomes inefficient when dealing with large query scopes. To address this issue, various index structures such as ROAD~\cite{ref10}, G-tree~\cite{ref11}, V-tree~\cite{ref9}, SIM$k$NN~\cite{ref14}, G*-tree~\cite{ref12}, TEN*-Index~\cite{ref7}, and ODIN~\cite{ref39} have been proposed to improve query efficiency. Nonetheless, these indexes are designed for static graphs with fixed edge weights and require frequent updates when applied to dynamic graphs, leading to substantial maintenance costs.

Considering the high overhead of maintaining a global index for dynamic graphs, we avoid employing a comprehensive index. Instead, we focus on INE-based exploration and propose a distributed $k$NN (D$k$NN) algorithm that employs a divide-and-conquer strategy to accelerate the search process. The algorithm partitions the graph into multiple small subgraphs. Starting from the subgraph containing the query point, D$k$NN performs a Dijkstra-based expansion. When the expansion reaches an adjacent subgraph, it initializes a separate search instance to continue the Dijkstra expansion within that subgraph, starting from the entry point. As exploration proceeds, additional search instances are created within each encountered subgraph until the $k$NNs are identified. Since these explorations across multiple subgraphs can run in parallel, this approach significantly speeds up the overall search.

Unlike centralized $k$NN algorithms, which can readily determine the true $k$NNs and identify when they have been fully discovered thanks to access to global search results, the D$k$NN approach cannot constantly fetch the search information from all search instances at any given moment. This leads to several challenges when searching for $k$NNs in road networks in a distributed manner: (1) A globally shortest path may traverse multiple subgraphs. Independent exploration within each subgraph cannot guarantee the preservation of shortest distances between vertices as they are on the entire graph, which is essential for accurately determining the $k$NN. (2) Each subgraph's search instance is unaware of objects discovered in other subgraphs, often leading to unnecessary expansion into adjacent subgraphs, even when further search is unlikely to yield closer objects. (3) Identifying the true $k$NNs relies on objects returned from multiple subgraphs, but without a global view, it is difficult to determine precisely when the $k$NNs have been found and that no closer objects will appear in the future.

Recently, some distributed approaches are proposed~\cite{ref21,ref22,ref23,ref24,ref25,ref27,ref41,ref26,ref29,ref30,ref31,ref32,ref28} to tackle the shortest path search on graphs. The presence of an explicit target point in these approaches allows for easy identification of when the exact result is determined, enabling timely termination of the distributed search process. 
However, these methods do not effectively address the aforementioned challenges in the context of $k$NN searches, where no explicit target exists, and the termination criteria are not as straightforward. 

To address these challenges, D$k$NN first establishes a candidate search scope based on the initially identified objects, ensuring this scope encompasses the true $k$NNs. Consequently, only the subgraphs intersecting with this candidate scope need to be explored, avoiding the excessive expansion on unpromising subgraphs. Within each subgraph, an entry-point-driven exploration mechanism is employed. Specifically, when an entry-point of a subgraph is reached, the corresponding search instance uses this entry point as a source to perform Dijkstra-based expansion within the subgraph, calculating the shortest distances from internal vertices to the query point. While a single exploration from one entry-point might not yield globally accurate shortest distances, multiple explorations from various entry-points can collectively refine and correct these inaccuracies. To prevent unnecessary exploration, a search instance initiates exploration only when an unvisited entry-point is reached, or if a previously visited entry-point is visited again with a new, shorter distance to the query point. Finally, D$k$NN implements a message-tracing mechanism that guarantees the correct identification of the true $k$NNs once all message exchanges between search instances have been completed. The main contributions of this work are as follows:

\begin{itemize}
    \item  We study the $k$NN problem over moving objects on dynamic road networks, explicitly accounting for the variable travel costs of road segments that occur in real-world scenarios.

    \item We introduce D$k$NN, a distributed $k$NN search algorithm that employs a divide-and-conquer strategy to parallelize exploration across multiple relevant subgraphs. This approach accelerates the search process, enables early identification of the exact $k$NNs, and helps avoid unnecessary exploration of unpromising subgraphs.
     
    \item D$k$NN is implemented on Apache Storm~\cite{ref40}, a distributed stream data processing platform. Extensive experiments conducted on multiple real-world road networks demonstrate the effectiveness and advantages of our approach compared to baseline methods.
\end{itemize}

This paper is organized as follows. Section~\ref{Section2} discusses related work, Section~\ref{Section3} presents some important preliminaries. Section~\ref{Section4} describes the D$k$NN algorithm. Section~\ref{Section5} presents the results of our experimental evaluation, and Section~\ref{Section6} concludes the paper.
\section{Related Work}\label{Section2}

In this section, we review the centralized approaches for solving $k$NN queries on road networks and distributed shortest path search algorithms on road networks.

\textbf{Centralized $k$NN search algorithms.} Methods like INE~\cite{ref2}, IER~\cite{ref3}, and IMA~\cite{ref4}  expand outward from the query vertex to explore nearby objects on the graph, moving from close to distant. However, when traversing regions with sparse objects, these methods often access numerous irrelevant vertices without associated objects for their lack of a macro-level understanding of graph structures, leading to unnecessary computational overhead. To mitigate this issue, various index structures~\cite{ref7,ref8,ref9,ref11,ref12,ref15,ref14,ref39,ref33} have been proposed to enhance query efficiency. For example, TOAIN~\cite{ref8} and TEN*-Index~\cite{ref7} utilize pre-built label trees for faster query processing. G-tree~\cite{ref11}, V-tree~\cite{ref9}, SIM$k$NN~\cite{ref14}, GLAD~\cite{ref15}, and ODIN~\cite{ref39} divide the graph into small partitions, and then avoid the exploration within the partitions that impossibly contain $k$NNs. KNN-Index~\cite{ref33} indexes $k$NNs for each vertex to reduce the search overhead. All these methods depend heavily on indexes, which need frequent updates in dynamic graphs, resulting in substantial maintenance costs.

\textbf{Distributed shortest path search algorithms.} Recently, some distributed shortest path search algorithms on graphs have been proposed~\cite{ref27,ref41,ref26,ref29,ref30,ref31,ref32,ref28}, but they are not directly applicable to the $k$NN search problem. In particular, some methods, such as ParaPLL~\cite{ref27} (an extension of PLL~\cite{ref29}) and DH-Index~\cite{ref41}, use pre-built indexes and face the same challenges as centralized methods. Other approaches, like those utilizing PowerGraph~\cite{ref30}, Pregel~\cite{ref31}, and GPS~\cite{ref32}, employ index-free shortest path algorithms but often struggle to deliver high query efficiency when handling large-scale searches. Similar to our approach, CANDS~\cite{ref28} divides the graph into subgraphs for parallel search; however, CANDS is primarily designed for single-source shortest path problems. Its query processing involves linear expansion from the query point to the target, while the distributed processing of $k$NN queries requires concurrent exploration across multiple subgraphs, a capability that exceeds the scope of CANDS.


\section{Preliminaries}\label{Section3}

\begin{definition}[Undirected graph]
The road network is modeled as an undirected  weighted graph $G=(V,E,W)$ consists of (1) a finite set of vertices $V$, (2) a set of edges $E \subseteq V \times V$, where $e_{i,j} \in E$ connects $v_i$ and $v_j$, and (3) a set of non-negative weights $W$. 

To account for varying weights of the road network, each snapshot of the road network is represented by such a graph at a specific time point. The $K$NN queries arriving during the interval between consecutive snapshots are processed using the current snapshot.
\end{definition}

\begin{definition}[Path, Path Distance]
Path $P(v_s,v_d)$ from the source vertex $v_s$ to the target vertex $v_d$ in graph $G$ is a sequence of vertices $(v_0=v_s,v_1...,v_i...,v_n=v_d)$ such that each $(v_{i-1},v_i) \in E$. We restrict to simple paths (no repeated vertices). Its distance is $D(P(v_s,v_d))=\sum_{i=1}^{n} w_{i-1,i}$. The shortest path distance from vertex $v_s$ to $v_d$, which is represented as $SD(v_s,v_d)$.


Our distributed search algorithm breaks down finding $k$NNs across the entire graph into parallel searches of partial $k$NNs within partitioned subgraphs. The subgraph is defined as follows.
\end{definition}

\begin{definition}[Subgraph]
A graph $SG=(V’,E’,W’)$ is a subgraph of the graph $G$ if (1) $V'\subseteq V$, (2) $E^{'} \subseteq E\land (e_{i,j} \in E', {v_i, v_j}\in V')$, and (3) $W'\subseteq W\land (w_{i,j} \in W'\rightarrow e_{i,j}\in E')$.

We use the METIS algorithm~\cite{ref19} to conduct graph partition. Partitioned subgraphs are connected via external edges between Border vertices.

\textbf{Border vertex.} For a vertex $v_i$ in subgraph $SG_p$, if the vertex $v_i$ has at least one adjacent vertex $v_j\in SG_q (q\neq p)$ ,then $v_i$ is called a border vertex of subgraph $SG_p$.

\textbf{External edge.} An edge $e (e\in E)$ is an external edge if it connects two border vertices in two different subgraphs.

\begin{figure}[t!]
    \vspace{-0.4cm}
    \centering
    \includegraphics[width=0.8\linewidth]{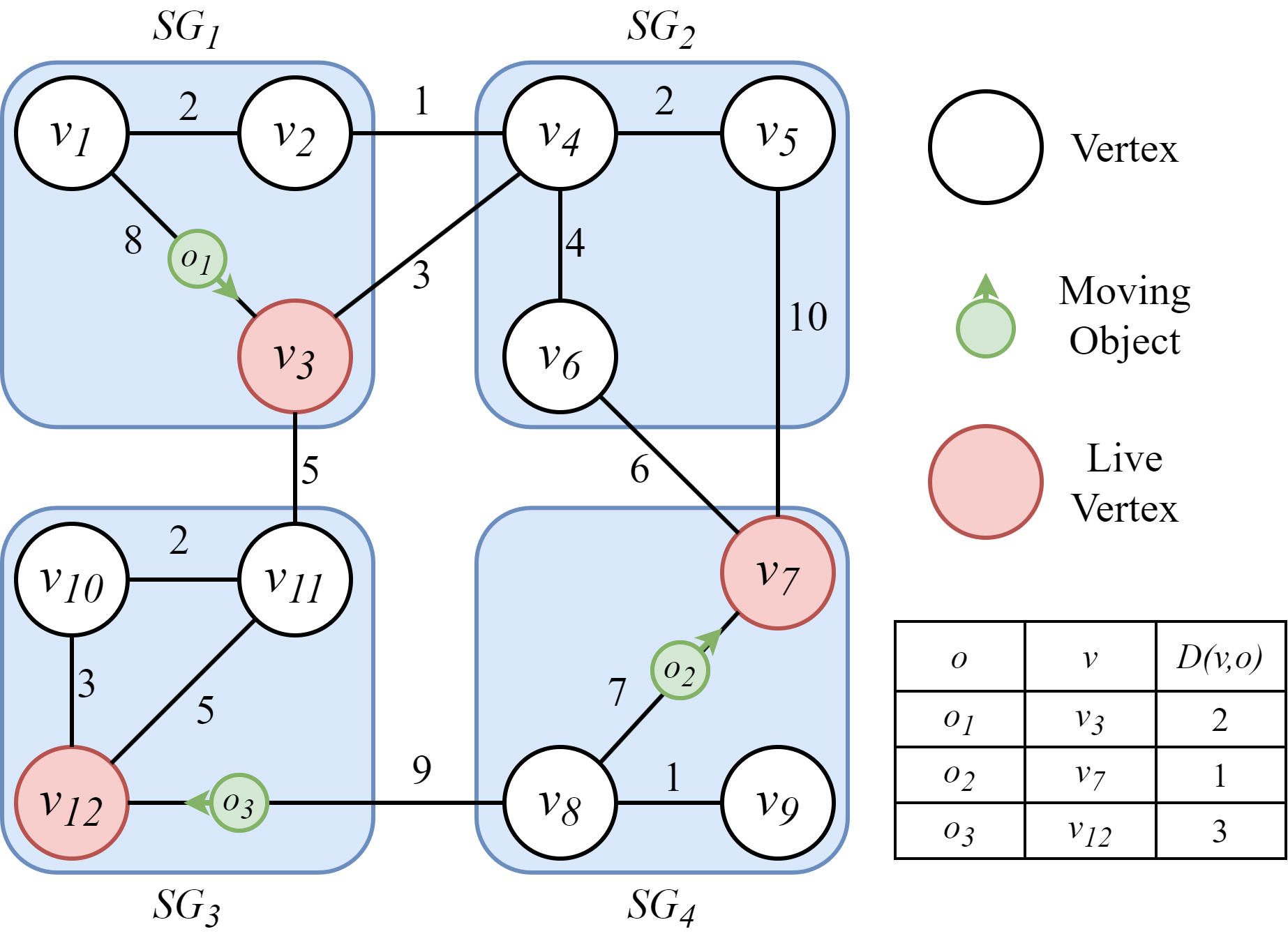}
    \caption{Example of subgraph partitioning with border vertices}
    \label{FIG1}
    \vspace{-0.5cm}
\end{figure}

Fig.~\ref{FIG1} illustrates an example of subgraph partitioning. The live vertices are highlighted in red, indicating the objects moving toward them.

\textbf{Live vertex.}
In our model, each moving object $o$ is associated with the next vertex $v_i$ it is traveling towards. Such a vertex $v_i$ is called a live vertex. For example, in Fig.~\ref{FIG1}, $o_1$ is moving towards $v_3$, and $o_2$ towards $v_7$; both $v_3$ and $v_7$ are considered live vertices.
The shortest distance from a vertex $v_d$ to a moving object $o$ associated with $v_i$ $(v_i \neq v_d)$ can be represented as $S D(v_d, o)=SD(v_d,v_i)+D(v_i,o)$ , where $D(v_i,o)$ represents the remaining travel cost for $o$ to reach $v_i$.
\end{definition}

\begin{definition}[$k$NN Query]
Given a graph $G=(V,E,W)$, a set of moving objects $O(|O|>k)$, and a $k$NN query $q(v_q,k)$, $k$NNs of $q$ refer to a set of objects $R(q)$, satisfying:
1) $|R(q)|=k$; 2) $R(q)\subseteq O$; and 3) $\forall o\in R(q),\forall o'\in (O-R(q)),SD(v_q,o) \leq SD(v_q,o')$.
\end{definition}
\section{D$k$NN Algorithm}\label{Section4}

\subsection{Overview of D$k$NN}

The D$k$NN algorithm adopts a divide-and-conquer principle to search $k$NNs on the graph in a distributed manner. It partitions the graph into multiple subgraphs, aiming to break down the global $k$NN search into smaller, localized searches within relevant subgraphs. For each query, exploration within each pertinent subgraph is conducted independently by separate search instances. These instances can operate concurrently, significantly reducing the overall computational cost and accelerating the query processing. 


When a $k$NN query arrives, D$k$NN locates the subgraph containing the query point and creates a search instance that performs an intra-subgraph exploration, discussed in Section~\ref{subsubsec:intra-subgraph-exploration}. This exploration starts from the query point and expands the search scope based on the Dijkstra's algorithm to discover the objects within this subgraph. When reaching a border vertex, this indicates the adjacent subgraph with connections to this border vertex also need to be investigated. Hence, a search instance will be initialized to explore the adjacent subgraph. After exploring all relevant subgraphs that can be determined using the approach presented in Section~\ref{query-range-pruning}, the true $k$NNs can be determined using the distributed query termination mechanism discussed in Section~\ref{subsubsec:termination}.  

\subsection{Key Components of D$k$NN}

\subsubsection{Intra-Subgraph Exploration}\label{subsubsec:intra-subgraph-exploration}
\indent

The intra-subgraph exploration is performed within each relevant subgraph by its designated search instance. When a search instance receives a query request from an adjacent subgraph via a border vertex, it designates that border vertex as the source and begins the intra-subgraph exploration. The query request message contains these attributes: query $id$ ($qid$), starting border vertex ($v_b$), distance from query vertex to $v_b$ ($dist$), originating subgraph ($from$). For each query, every subgraph maintains an upper bound $\epsilon$, representing the current best $k$NN distance of the subgraph.
These attributes guide the search instance in making search decisions.

The Intra-Subgraph Exploration (ISE) procedure executed by each search instance is detailed in Algorithm~1. ISE employs the Dijkstra's algorithm to explore the subgraph from the starting border vertex $v_b$. In this process, the locally optimal distances from $v_q$ to each encountered object via $v_b$ can be determined. Moreover, on the current snapshot of the graph, ISE maintains a reusable local Dijkstra cache of shortest path distances originating from the entry border vertex $v_b$ to all other vertices within the subgraph (lines 5-6). This cache, once computed, can serve subsequent query messages arriving at the same entry point. The distance from the query vertex $v_q$ to an internal vertex is then derived by combining the incoming distance $D(v_q, v_b)$ with the corresponding distance retrieved from this local Dijkstra cache. This derived distance is then stored in a query-specific query distance cache that tracks the best-known distances from $v_q$ to the live and border vertices within the subgraph. When dealing with the subsequent query requests, if ISE discovers shorter distances from $v_b$ to the live and border vertices via other starting border vertices, it updates the query distance cache accordingly. ISE then adds the moving objects from these updated live vertices to the local candidate result queue $Q_l$ (lines 7-15). Once all queries related to the current snapshot have been processed, the search instance clears these caches before moving on to the next graph snapshot.

\begin{algorithm}[h]
  \caption{Intra-Subgraph Exploration (ISE)}
  \LinesNumbered
  \KwIn{Query $id$ $qid$, border vertex $v_b$, path distance from the query vertex to $v_b$ $dist$, subgraph $SG_p$, query upper bound $\epsilon$\;}
  \KwOut{Moving object set $Q_l$, query message set $M$\;}
  get $D(v_q,v_b)$ from the query distance cache in $SG_p$\;
  \If{$dist>\epsilon$ or $dist>D(v_q,v_b)$}{
    return $\emptyset$,$\emptyset$\;
  }
  $D(v_q,v_b)=dist$\;
  \If{$v_b$ is not visited}{
    Execute Dijkstra$(v_b,SG_p)$; Set $v_b$ as visited\;
  }
  \ForEach{live vertex $v_i$ of $SG_p$}{
    get $D(v_q,v_i)$ from the query distance cache in $SG_p$\;
    get $D(v_b,v_i)$ from Dijkstra$(v_b,SG_p)$\;
    \If{$dist+D(v_b,v_i)\leq D(v_q,v_i)$}{
      $D(v_q,v_i)=dist+D(v_b,v_i)$\;
      \ForEach{moving object $o$ of $v_i$}{
        \If{$D(v_q,v_i)+D(v_i,o)<\epsilon$}{
          $o.dist=D(v_q,v_i)+D(v_i,o)$\; Add moving object $o$ to set $Q_l$\;
        }
      }
    }
  }
  Initialize vertex set $V_n$ to store border vertices\;
  \ForEach{each border vertex $v_j$ of $SG_p$}{
    \ForEach{each adjacent vertex $v_b^{'}$ of $v_j$}{
      \If{$v_b^{'}\notin SG_p$}{
        \If{$D(v_q,v_b^{'})\geq dist+D(v_b,v_j)+D(v_j,v_b^{'})$}{
          $D(v_q,v_b^{'})=dist+D(v_b,v_j)+D(v_j,v_b^{'})$\; Add $v_b^{'}$ to set $V_n$\;
        }
      }
    }
  }
  \ForEach{each vertex $v_b^{'}$ in $V_n$}{
    Generate query message $msg=<qid, v_b^{'}, D(v_q, v_b^{'}), from=SG_p>$ and add it to set $M$\;
  }
  \Return $Q_l$, $M$\;
\end{algorithm}

Upon completing its local exploration, the subgraph sends corresponding messages to adjacent subgraphs and the query unit to facilitate search expansion and result integration. The exploration identifies and stores any border vertex whose shortest distance from the query point has been updated (lines 16-22); the algorithm then generates corresponding query request message $msg$ based on these border vertices and their connected subgraphs, and adds these messages to queue $M$ (lines 23-24). Upon search completion, the algorithm merges $msg$ targeting the same subgraph and sends them to the corresponding subgraph, reducing inter-subgraph communication overhead. Meanwhile, the local candidate result queue $Q_l$ is sent to the query unit for global result integration.



\subsubsection{Query Range Pruning}\label{query-range-pruning}
\indent

While intra-subgraph exploration efficiently finds local candidates, it generates requests to explore all adjacent subgraphs. 
In the D$k$NN algorithm, the exploration expands messages between adjacent subgraphs. This distributed propagation of search tasks can lead to asynchrony in the awareness of the global query upper bound, $\epsilon$. A subgraph might operate with an outdated, overly large $\epsilon$, causing query range overflow.


\begin{figure}[htbp]
    \vspace{-0.3cm}
    \centering
    \includegraphics[width=0.95\linewidth]{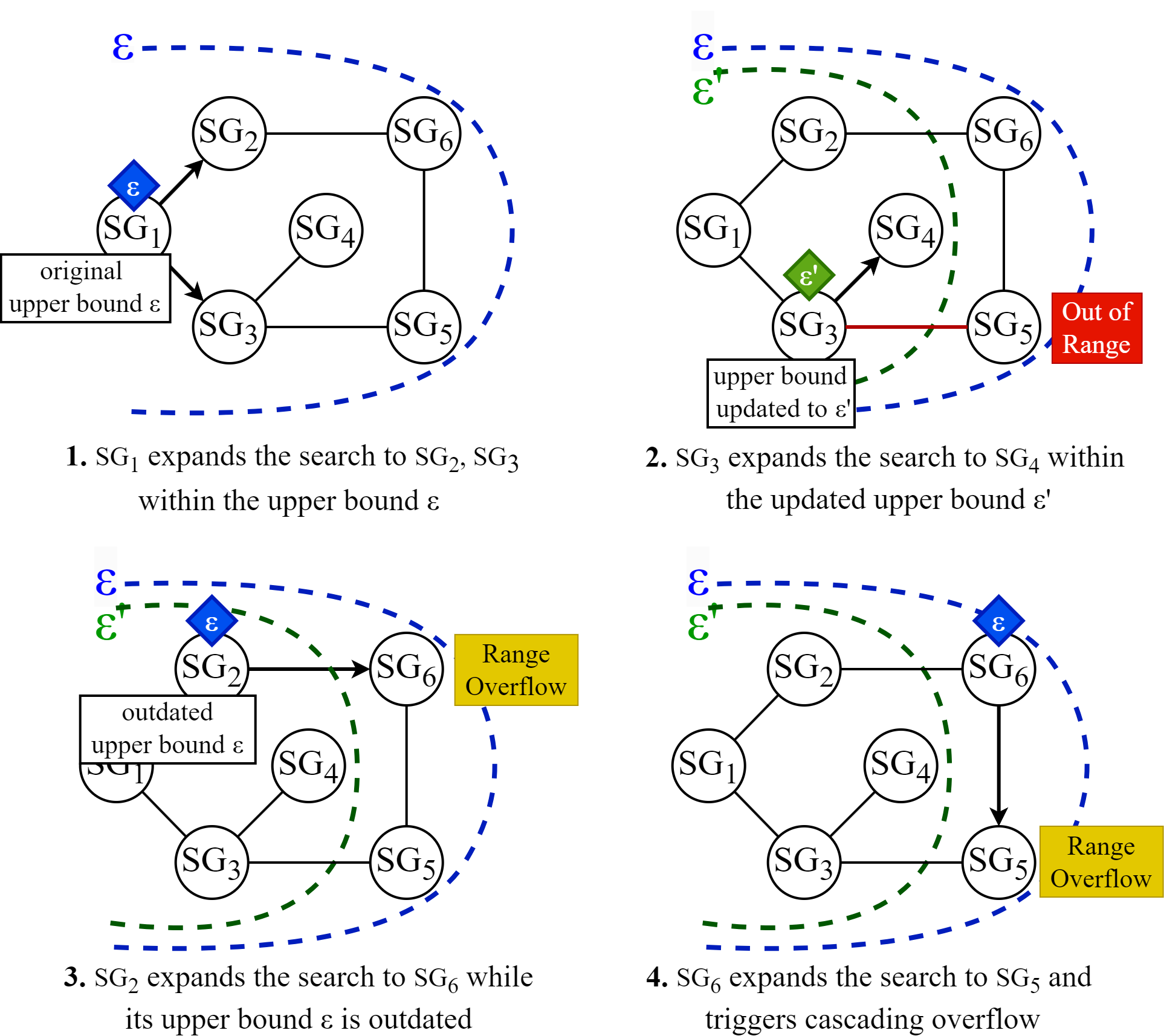}
    \caption{Example of query range overflow}
    \label{FIG2}
    \vspace{-0.7cm}
\end{figure}

For example, in the situation shown in Fig.~\ref{FIG2}, upon completion of its local search, subgraph $SG_1$ initiates parallel explorations into adjacent subgraphs $SG_2$ and $SG_3$, propagating its current upper bound $\epsilon$ to both. If the search in $SG_3$ finishes first and updates the upper bound to a tighter value $\epsilon'$, any subsequent expansion to a subgraph like $SG_5$, which satisfies the old bound $\epsilon$ but not $\epsilon'$, is correctly pruned. However, the search instance in $SG_2$ may still be operating with the outdated bound $\epsilon$. 
This lack of synchronization can lead it to unnecessarily expand into regions like $SG_6$ that lie outside the valid scope defined by $\epsilon'$.

To solve this problem, D$k$NN employs a hybrid pruning strategy, which accurately shrinks the query range through hybrid pruning using Euclidean distance and graph distance. The query manages moving objects with a priority queue in the query unit. When $k$ moving objects are collected, the distance of the $k$-th object is taken as the new upper bound $\epsilon$, ensuring that $\epsilon$ is always the maximum distance threshold of the current optimal solution. Utilizing the property that Euclidean distance $\leq$ graph distance, the query unit precomputes the minimal Euclidean distance to any border vertex of $E(v_q,SG_p)(SG_p\in SG)$ from the query vertex to each subgraph. If $E(v_q,SG_p)>\epsilon$, the subgraph is directly pruned to avoid invalid expansion.

Specifically, we can calculate the Euclidean distance between the query point and each border vertex in a subgraph based on their coordinates, and take the minimum value among all these Euclidean distances as the $E(v_q, SG_p)$ of that subgraph. Based on the Euclidean distance of each subgraph, the query unit maintains the current set of effective subgraphs $SG'$. When the result queue collects $k$ moving objects for the first time, the algorithm adds all subgraphs with $E(v_q, SG_p) \leq \epsilon$ to $SG'$; when $\epsilon$ is updated to a smaller value, the subgraphs in $SG'$ are filtered in descending order of their Euclidean distances. After pruning, the new $\epsilon$ is only distributed to the subgraphs in $SG'$, and the $\epsilon$ of the pruned subgraphs is set to 0 to terminate their query processing.

\begin{figure}[hbt!]
    \centering
    \includegraphics[width=0.9\linewidth]{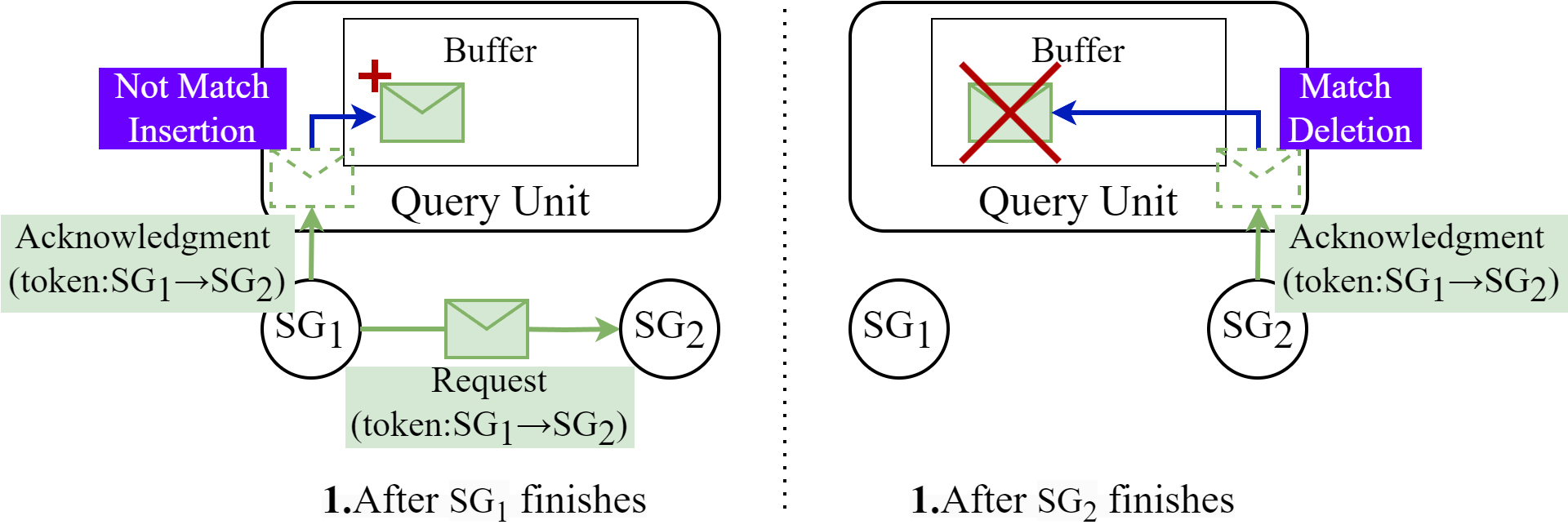}
    \caption{Example of two-way matching procedure}
    \label{FIG3}
    \vspace{-0.6cm}
\end{figure}

\subsubsection{Distributed Query Termination Mechanism}\label{subsubsec:termination}
\indent

In D$k$NN, a central query unit manages global query coordination, including result collection, subgraph pruning, and termination. To reliably track progress in a distributed setting, D$k$NN employs a centralized message acknowledgment mechanism based on a two-way token matching principle.

When a search instance in a subgraph $SG_p$, dispatches a query message to an adjacent subgraph $SG_q$, it synchronously notifies the query unit. The query unit then generates a unique acknowledgment token, $\tau = \langle q, SG_p \to SG_q \rangle$, and inserts it into a query-specific acknowledgment buffer $\mathcal{B}_q$. This token acts as a placeholder signifying an outstanding exploration path.
Upon completing its local search and any subsequent expansions, the receiving subgraph $SG_q$ is responsible for ``closing'' this path. It does so by sending a corresponding acknowledgment message token $\tau$, back to the query unit. The query unit then attempts to match this incoming acknowledgment with a token in the buffer. If a matching token $\tau' \in \mathcal{B}_q$ is found (where $\tau' = \tau$), it is removed from the buffer, indicating that the exploration path has been fully resolved.

For instance, as illustrated in Fig.~\ref{FIG3}, when the search in $SG_1$ expands to $SG_2$, a token $\langle q, SG_1 \to SG_2 \rangle$ is added to the buffer. After the search instance in $SG_2$ finishes its processing, it sends an acknowledgment for this specific token back to the query unit. The returning acknowledgment is then matched with the token in the buffer, and the token is removed.

A query terminates if and only if its acknowledgment buffer becomes empty. We formalize this guarantee with the following theorems.

\begin{theorem}
\label{thm:completeness}
The acknowledgment buffer $\mathcal{B}_q$ is empty if and only if all dispatched messages have been fully processed.
\end{theorem}
\begin{proof}
By definition, a token exists in $\mathcal{B}_q$ for each dispatched but unacknowledged message. Thus, $\mathcal{B}_q = \emptyset$ directly implies that the set of query messages equals the set of acknowledged messages, meaning all tasks have concluded. Conversely, if an unprocessed message existed, its token would persist in $\mathcal{B}_q$, contradicting the empty-buffer condition.
\end{proof}

\begin{theorem}
\label{thm:correctness}
Termination under the condition $\mathcal{B}_q=\emptyset$ guarantees that the result set $R(q)$ contains the globally optimal $k$NNs.
\end{theorem}

\begin{proof}
Theorem~\ref{thm:completeness} ensures search completeness upon termination. We prove optimality by contradiction. Let $\epsilon = \max_{o \in R(q)} SD(v_q, o)$ be the final upper bound. Assume a better candidate object $o'_{\text{cand}} \notin R(q)$ exists with $SD(v_q, o'_{\text{cand}}) < \epsilon$.
To discover $o'_{\text{cand}}$, a search path must exist where every vertex $v$ on the path satisfies $SD(v_q, v) < \epsilon$. According to D$k$NN's protocol, such a path would be fully explored, as it always lies within the bound $\epsilon$. This exploration would generate a sequence of tokens in $\mathcal{B}_q$. Since the algorithm has terminated ($\mathcal{B}_q = \emptyset$), all these tokens must have been resolved, implying this path was fully explored. A full exploration would have identified $o'_{\text{cand}}$. Given its superior distance, it would have been included in $R(q)$, which contradicts our initial assumption. Thus, no such better candidate can exist.
\end{proof}

\subsection{Time Complexity Analysis}

\subsubsection{Complexity of a Single Query}
The time complexity of D$k$NN for a single query is analyzed based on three key parameters: the original graph size $(|V|, |E|)$, the number of subgraphs $m$, and the number of subgraphs accessed $D$.
Each subgraph has average size $|V|/m$ vertices and $|E|/m$ edges. The local search within each subgraph uses Dijkstra's algorithm with complexity $O((|V|/m + |E|/m) \log (|V|/m))$.
The total cost for a single query is:
\begin{equation}
T_{\text{single}} = O\left(D \cdot (|V|/m + |E|/m) \log (|V|/m)\right)
\end{equation}

In the worst case where $D = m$, this simplifies to $O((|V| + |E|) \log |V|)$, matches the asymptotic cost of centralized approaches. The pruning strategy minimizes $D$ for practical efficiency.

\subsubsection{Complexity Under Parallel Execution}
For $R$ queries processed concurrently by $P$ distributed units:
\begin{equation}
T_{\text{total}} = O\left((R/P) \cdot D \cdot (|V|/m + |E|/m) \log (|V|/m)\right)
\end{equation}

D$k$NN achieves scalability through effective pruning (reducing $D$) and parallel processing (utilizing $P$ units).
\section{Experiments}\label{Section5}

\subsection{Implementation on Apache Storm}

\begin{figure}[t!]
    \vspace{-0.3cm}
    \centering
    \includegraphics[width=0.8\linewidth]{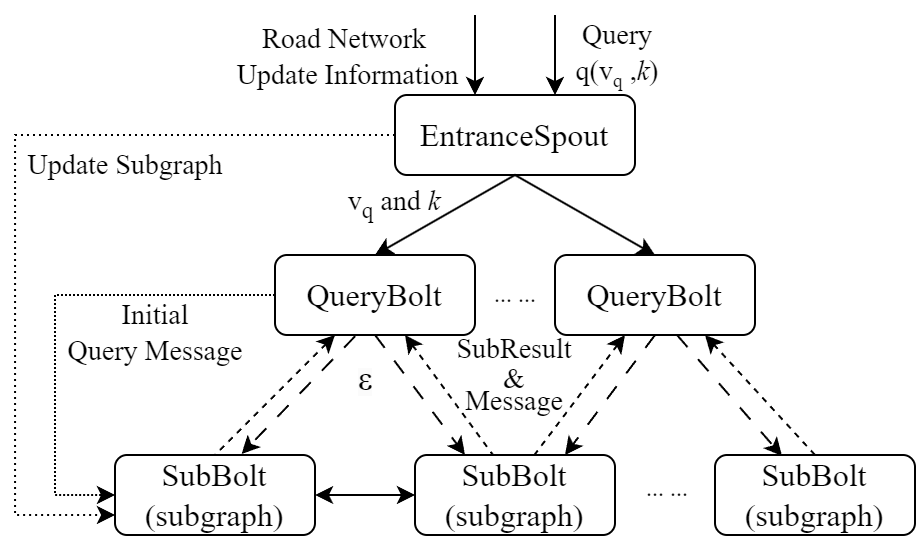}
    \caption{D$k$NN deployment on Apache Storm}
    \label{FIG4}
    \vspace{-0.6cm}
\end{figure}

To evaluate the performance of the D$k$NN algorithm, this study implements it based on Storm~\cite{ref40}, a distributed stream data processing framework. In line with Storm's paradigm, D$k$NN is designed as a directed acyclic graph (DAG) Topology where Spouts (data source nodes) and Bolts (logical computing nodes) are connected via streams (unbounded sequences of tuples). Spouts provide data for the Topology, while each Bolt processes tuples according to user-defined logic. 

The deployment of D$k$NN on Apache Storm (Fig.~\ref{FIG4}) consists of an EntranceSpout, SubBolts, and a QueryBolt. The EntranceSpout receives edge weight updates and $k$NN queries, maintains graph partitioning, and dispatches updates to SubBolts. Each SubBolt manages one partitioned subgraph. The QueryBolt receives local $k$NN results and confirms query termination. Query processing follows two steps:

1. The EntranceSpout encapsulates query $q(v_q, k)$ into a tuple for the QueryBolt, which then locates $v_q$'s subgraph and dispatches a query message to the corresponding SubBolt.

2. The target SubBolt traverses from $v_q$, generates messages for border vertices to adjacent subgraphs, and distributes them. Receiving SubBolts perform intra-subgraph searches, extend to adjacent subgraphs, and send found objects and acknowledgments to the QueryBolt. The QueryBolt maintains a priority queue for results, sets the $k$-th object‘s distance as $\epsilon$, prunes subgraphs, and synchronizes $\epsilon$. It uses an acknowledgment buffer to determine termination: when empty, it outputs final $k$NN results; otherwise, the query continues.

\subsection{Experimental Setup}
All algorithms discussed are implemented on Storm. Storm runs on a computer equipped with a 2.2GHz CPU and 16GB of memory.
For evaluation, four real-world road networks from the USA are used. The information
of each road network is shown in Table~\ref{table}. Among these, the medium-sized FLA dataset is set as the default for experiments. 

\begin{table}[!t]   
    \caption{Statistics on the Road Network Datasets}
    \vspace{-0.3cm}
    \begin{center}
        \begin{tabular}{c|c|c}
            \hline
            \textbf{Dataset} & \textbf{\#Vertices} & \textbf{\#Edges}\\
            \hline
                NY & 264,346 & 733,846 \\
                COL & 435,666 & 1,057,066 \\
                FLA & 1,070,376 & 2,712,798 \\
                CAL & 1,890,815 & 4,657,742 \\
            \hline
        \end{tabular}
    \label{table}
    \end{center}    
    \vspace{-0.7cm}
\end{table}

We implement Dijkstra~\cite{ref1}, TEN*-Index~\cite{ref7}, H2H~\cite{ref42}, and SIM$k$NN~\cite{ref14} as baseline methods. For fairness, multiple instances of each baseline algorithm are implemented in Storm to enable their parallel execution. 

\begin{figure*}[htbp]
    \vspace{-0.5cm}
    \centering
    \subfigure[]{
        \begin{minipage}{0.23\linewidth}
            \centering
            \includegraphics[scale=0.42]{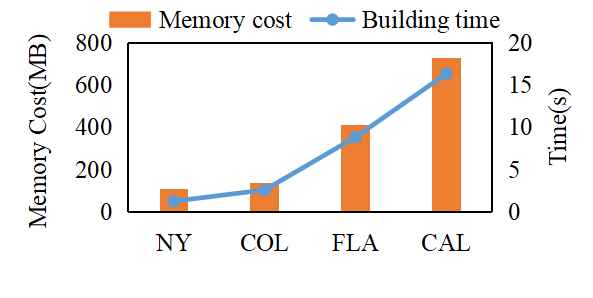}
        \end{minipage}}
    \subfigure[]{
        \begin{minipage}{0.23\linewidth}
            \centering
            \includegraphics[scale=0.42]{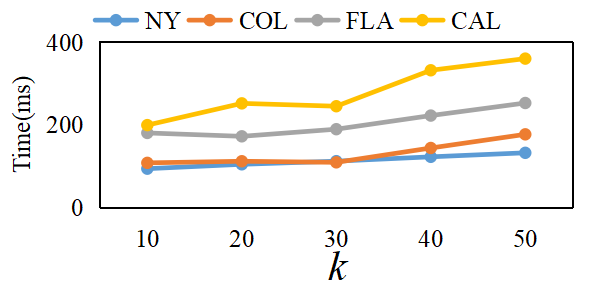}
        \end{minipage}}
    \subfigure[]{
        \begin{minipage}{0.23\linewidth}
            \centering
            \includegraphics[scale=0.42]{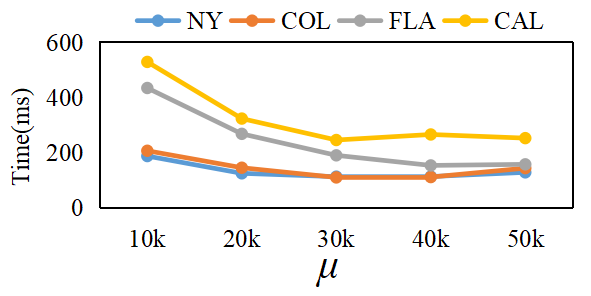}
        \end{minipage}}
    \subfigure[]{
        \begin{minipage}{0.23\linewidth}
            \centering
            \includegraphics[scale=0.42]{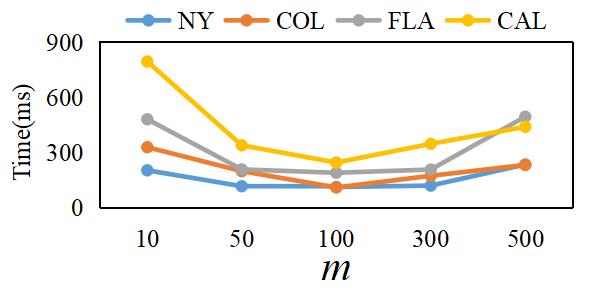}
        \end{minipage}}
    \subfigure[]{
        \begin{minipage}{0.23\linewidth}
            \centering
            \includegraphics[scale=0.42]{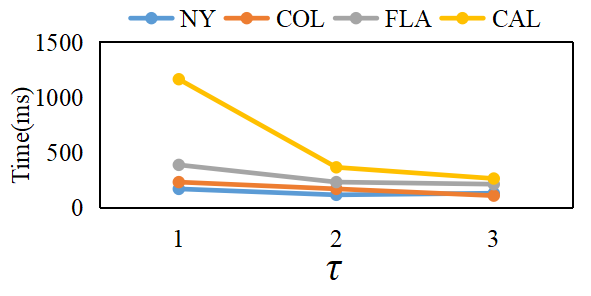}
        \end{minipage}}
    \subfigure[]{
        \begin{minipage}{0.23\linewidth}
            \centering
            \includegraphics[scale=0.42]{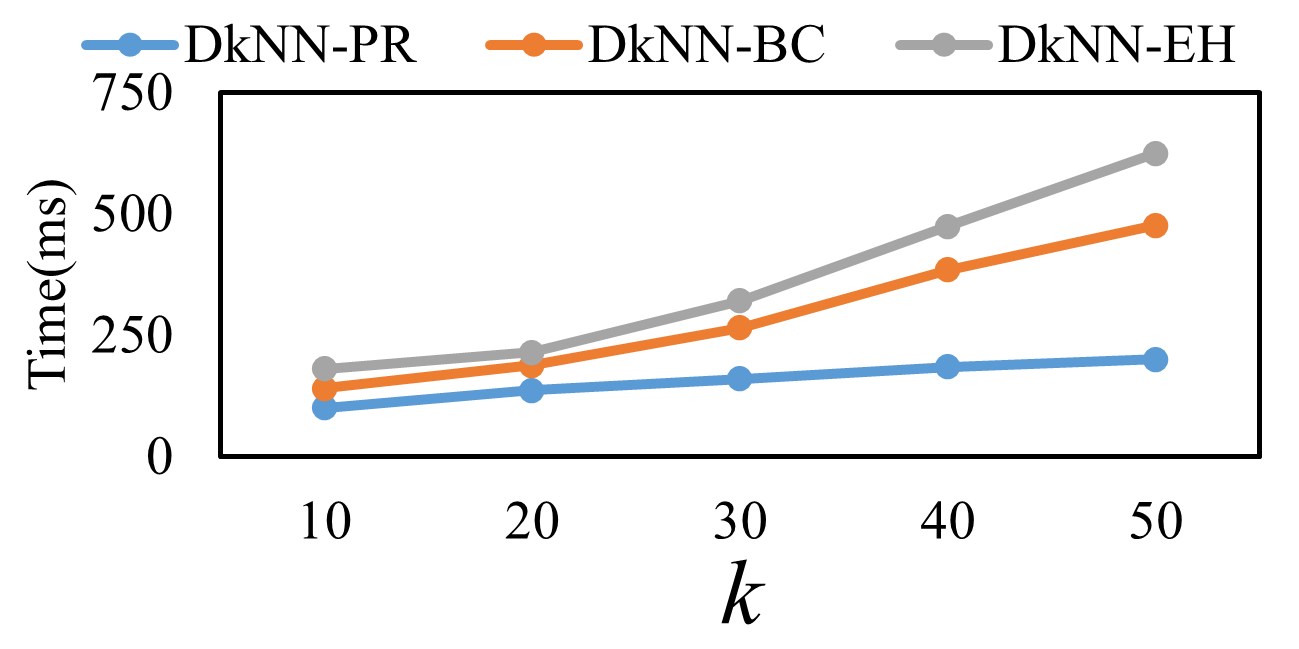}
        \end{minipage}}
    \subfigure[]{
        \begin{minipage}{0.23\linewidth}
            \centering
            \includegraphics[scale=0.42]{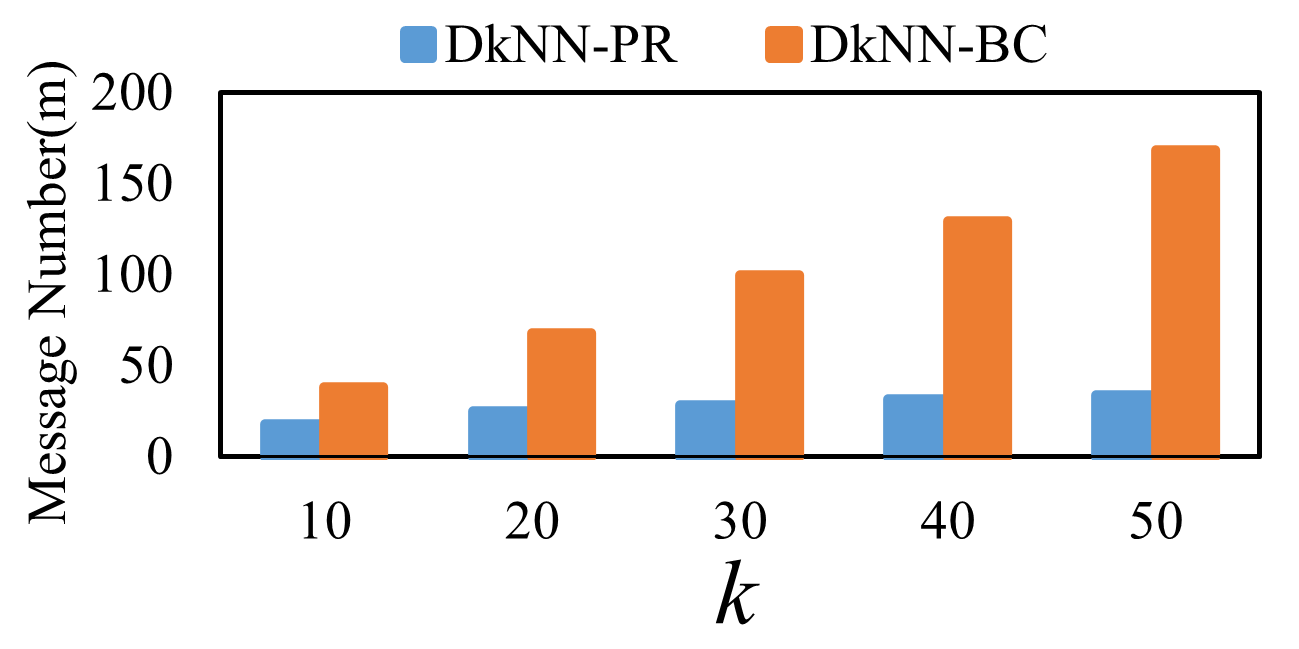}
        \end{minipage}}
    \subfigure[]{
        \begin{minipage}{0.23\linewidth}
            \centering
            \includegraphics[scale=0.42]{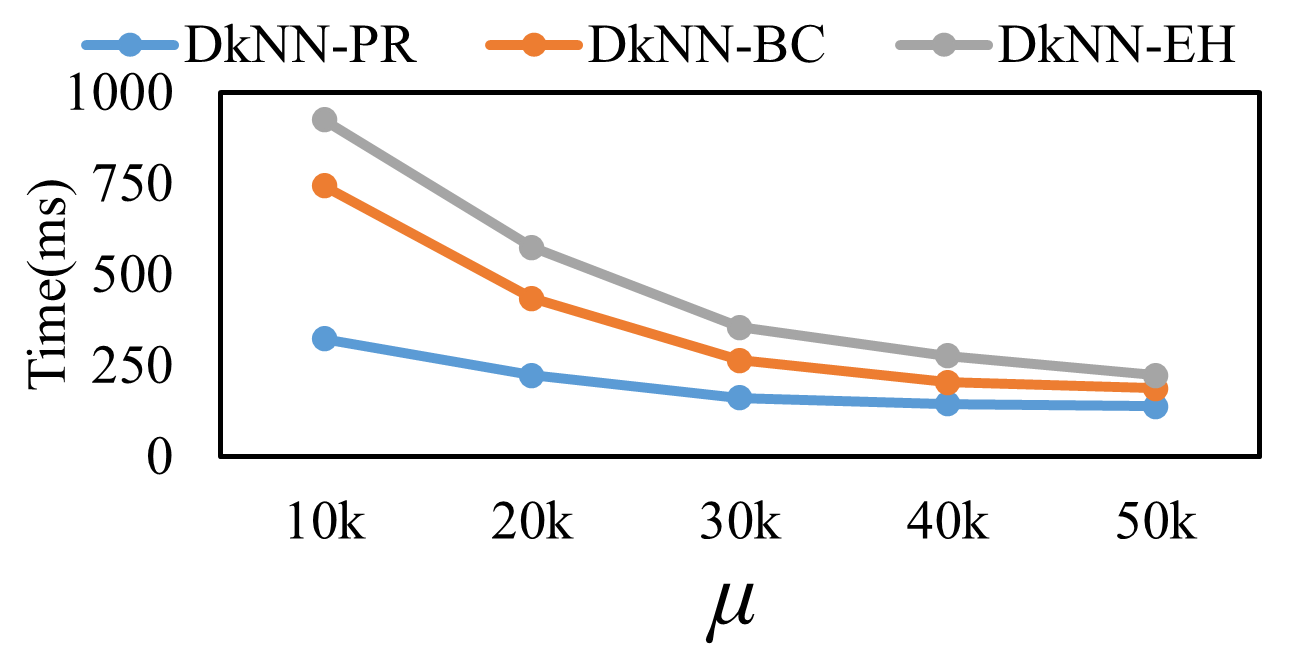}
        \end{minipage}}
    \subfigure[]{
        \begin{minipage}{0.23\linewidth}
            \centering
            \includegraphics[scale=0.42]{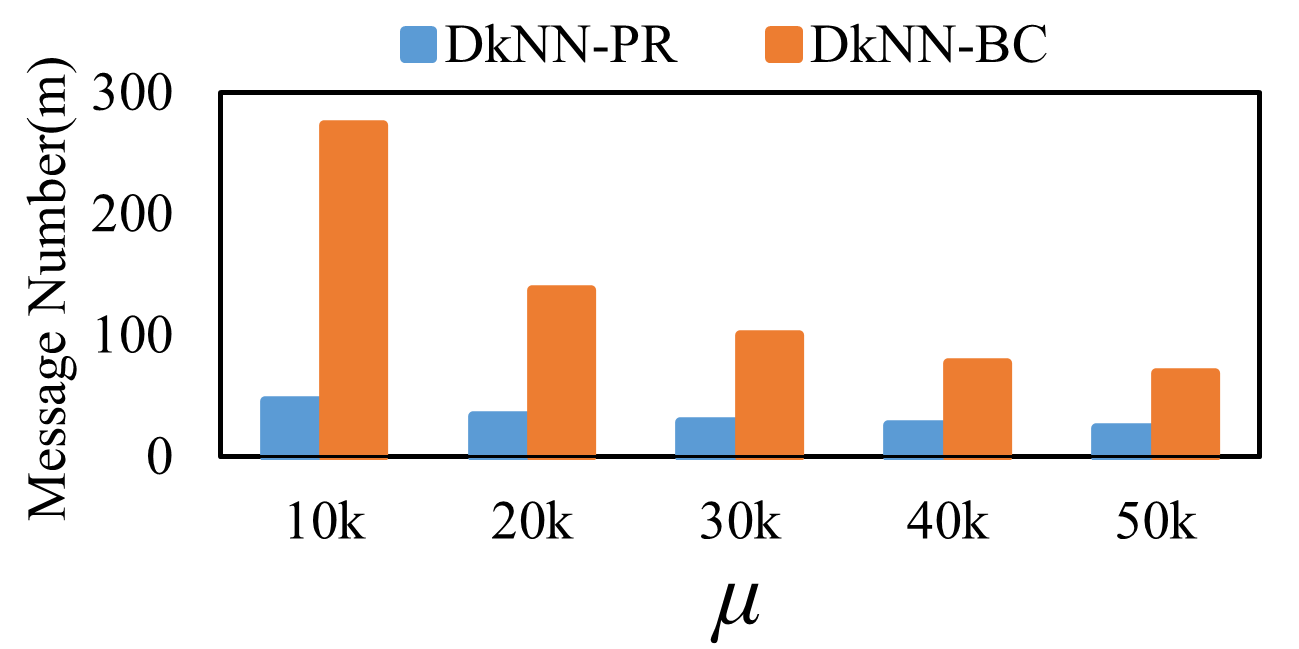}
        \end{minipage}}
    \subfigure[]{
        \begin{minipage}{0.23\linewidth}
            \centering
            \includegraphics[scale=0.42]{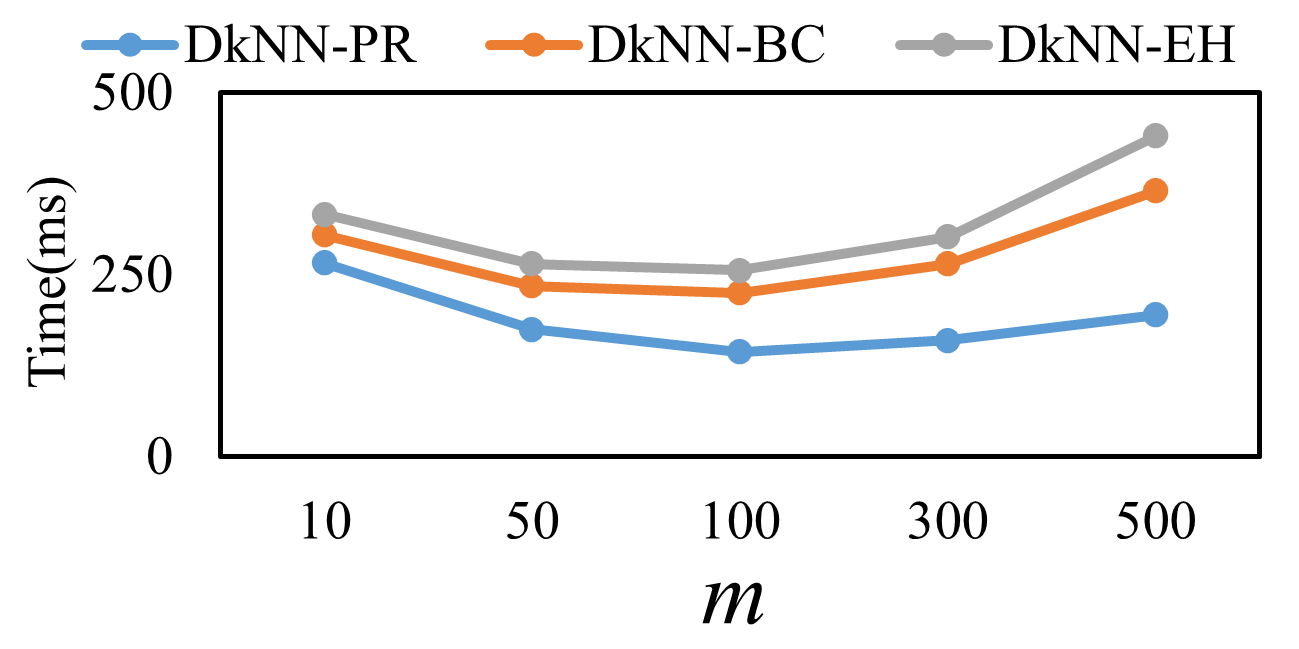}
        \end{minipage}}
    \subfigure[]{
        \begin{minipage}{0.23\linewidth}
            \centering
            \includegraphics[scale=0.42]{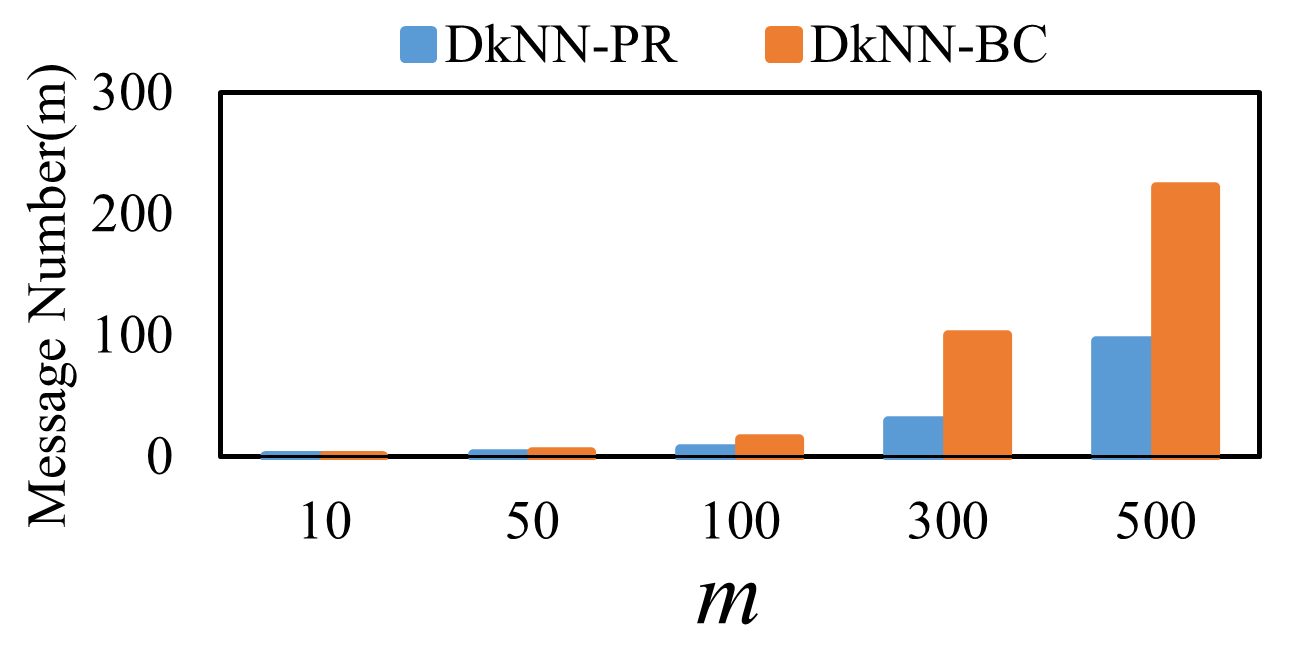}
        \end{minipage}}
    \subfigure[]{
        \begin{minipage}{0.23\linewidth}
            \centering
            \includegraphics[scale=0.42]{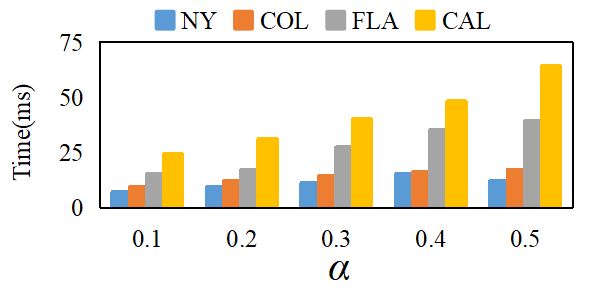}
        \end{minipage}}
    \vspace{-0.3cm}
    \caption{Evaluation of D$k$NN}    
    \label{FIG5}    
\end{figure*}

\subsection{Performance Evaluation of D$k$NN}

\textbf{Preprocessing Cost.} 
Fig.~\ref{FIG5}(a) presents the preprocessing cost of D$k$NN on four real-world road networks: NY, COL, FLA, and CAL. Preprocessing time and resource consumption increase with graph size, due to the growing number of vertices and edges that must be partitioned.

\textbf{Query Efficiency.}
We evaluate the query performance of D$k$NN under various parameter settings, as illustrated in Fig.~\ref{FIG5}(b)-(e). We analyze the impact of four key parameters on query efficiency: the number of nearest neighbors $k$, the number of moving objects $\mu$, the number of subgraphs $m$, and the number of threads $\tau$. Unless otherwise specified, the default values for these parameters are set to $k=30$, $\mu=30000$, $m=300$, and $\tau=20$, respectively.

As shown in Fig.~\ref{FIG5}(b), query time increases with $k$, which is attributed to the expansion of the search range. 
Fig.~\ref{FIG5}(c) reveals that query time initially decreases and then stabilizes as $\mu$ increases. A higher $\mu$ value increases the density of moving objects within subgraphs, reducing the need to access multiple subgraphs. When $\mu$ is sufficiently large, the $k$NNs can often be found within the local subgraph, leading to stable query times. 
The impact of the number of subgraphs $m$ is shown in Fig.~\ref{FIG5}(d). Query time first decreases and then increases with $m$. Small $m$ values limit parallelism and increase traversal time per subgraph. As $m$ grows, parallelism improves and query time drops. However, excessively large $m$ leads to higher inter-subgraph communication overhead, degrading performance. 
Fig.~\ref{FIG5}(e) illustrates the effect of thread $\tau$ on query time. Increasing $\tau$ enhances parallelism by distributing subgraph processing across more threads, thereby improving query efficiency and demonstrating the scalability of D$k$NN.

\textbf{Ablation Study.} 
An ablation study on the FLA dataset evaluates the efficacy of the proposed query range pruning strategy (Fig.~\ref{FIG5}(f)–(k)). We compare the query time and the number of messages generated by three variants: D$k$NN-EH (no pruning), D$k$NN-PR (with pruning in Section~\ref{query-range-pruning}), and D$k$NN-BC (broadcasting bounds to all subgraphs). 
Results show that D$k$NN-PR exhibits the lowest query time and generated messages across all parameter settings ($k$, $\mu$, and $m$). The pruning strategy effectively limits search expansion and bound propagation to relevant subgraphs, whereas D$k$NN-BC generates a significantly higher volume of messages and D$k$NN-EH has a higher computational cost. This indicates that the proposed pruning method achieves a balance between computational efficiency and communication overhead. 

\textbf{Update Efficiency.} 
We evaluate D$k$NN's ability to handle dynamic edge weight changes in Fig.~\ref{FIG5}(l), a key metric for performance in dynamic environments. The parameter $\alpha$ controls the proportion of updated edges. While the update time increases with a larger $\alpha$, D$k$NN's update cost remains on the millisecond scale. This efficiency stems from its index-free design. Instead of performing complex structural maintenance like index-based methods, D$k$NN's adaptation to a new graph snapshot is lightweight. 

\begin{figure*}[htbp]
    \centering  
    \subfigure[]{
        \begin{minipage}{0.23\linewidth}
            \centering
            \includegraphics[scale=0.42]{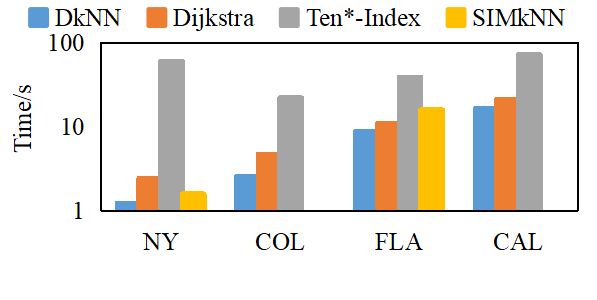}
        \end{minipage}}
    \subfigure[]{
        \begin{minipage}{0.23\linewidth}
            \centering
            \includegraphics[scale=0.42]{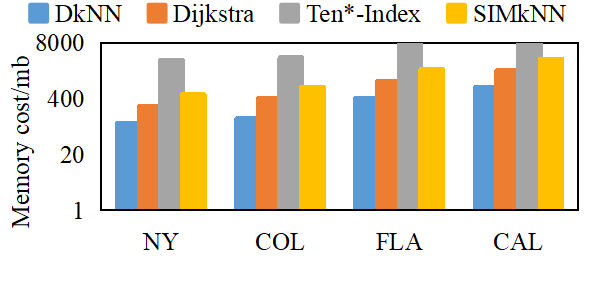}
        \end{minipage}}
    \subfigure[]{
        \begin{minipage}{0.23\linewidth}
            \centering
            \includegraphics[scale=0.42]{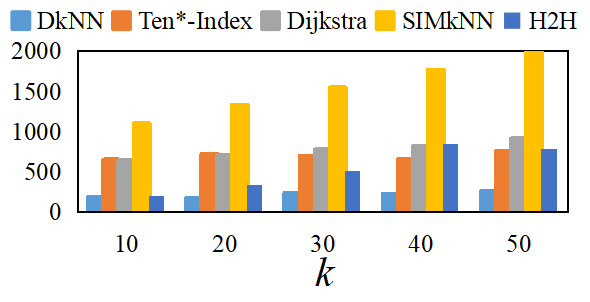}
        \end{minipage}}
    \subfigure[]{
        \begin{minipage}{0.23\linewidth}
            \centering
            \includegraphics[scale=0.42]{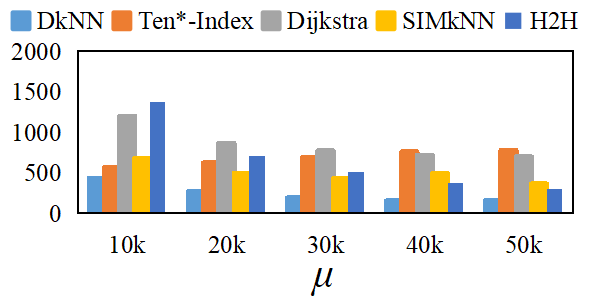}
        \end{minipage}}
    \subfigure[]{
        \begin{minipage}{0.23\linewidth}
            \centering
            \includegraphics[scale=0.42]{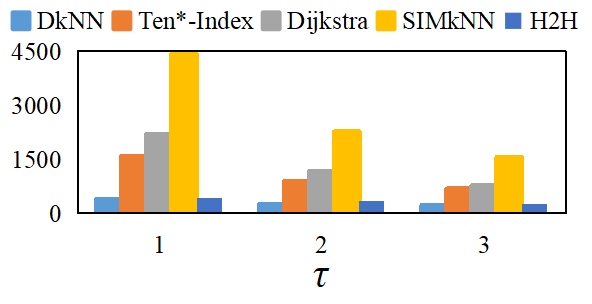}
        \end{minipage}}
    \subfigure[]{
        \begin{minipage}{0.23\linewidth}
            \centering
            \includegraphics[scale=0.42]{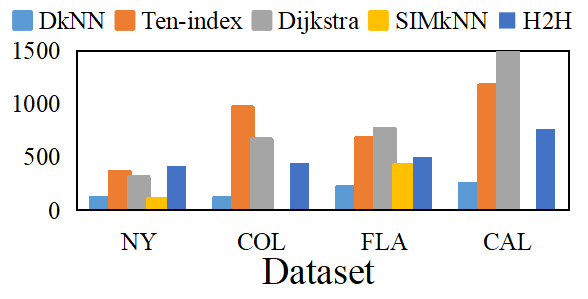}
        \end{minipage}}
    \subfigure[]{
        \begin{minipage}{0.23\linewidth}
            \centering
            \includegraphics[scale=0.42]{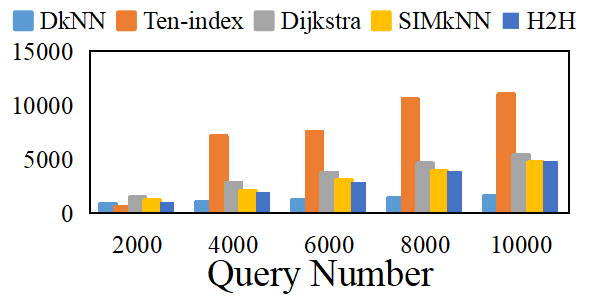}
        \end{minipage}}
    \subfigure[]{
        \begin{minipage}{0.23\linewidth}
            \centering
            \includegraphics[scale=0.42]{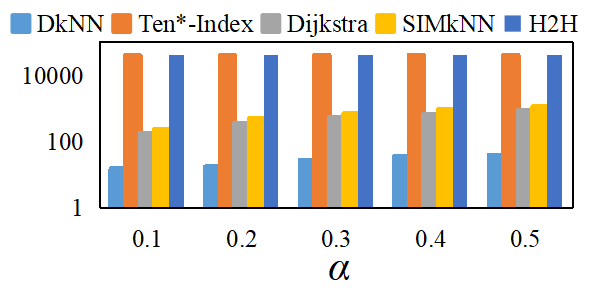}
        \end{minipage}}
    \vspace{-0.3cm}
    \caption{Comparison with Baselines}
    \label{FIG6}    
    \vspace{-0.5cm}
\end{figure*}

\subsection{Comparison with Baselines}

We conduct a comprehensive evaluation comparing D$k$NN against four baseline methods—Dijkstra, TEN*-Index, SIM$k$NN, and H2H—in terms of preprocessing overhead, query efficiency under various parameter settings, scalability across graphs of different sizes, and update performance under varying update intensities $\alpha$. 

\textbf{Preprocessing Cost Comparison.} 
Fig.~\ref{FIG6}(a)–(b) summarizes the preprocessing costs of all algorithms. As shown in Fig.~\ref{FIG6}(a), TEN*-Index and SIM$k$NN incur substantial preprocessing time due to tree decomposition and grid-based indexing, respectively. In contrast, D$k$NN only performs graph partitioning without building a full index, leading to lower preprocessing overhead. Dijkstra, which avoids both indexing and partitioning, exhibits the shortest preprocessing time. Fig.~\ref{FIG6}(b) further shows that the space consumption trends align with time costs: index-based baselines require significantly more storage than D$k$NN’s lightweight partitioning mechanism.

\textbf{Query Efficiency Comparison.} 
Fig.~\ref{FIG6}(c)–(g) present the query performance under different parameters. In Fig.~\ref{FIG6}(c), Dijkstra exhibits the longest query time, which grows considerably as $k$ increases due to its exhaustive vertex-wise expansion. SIM$k$NN improves upon Dijkstra by pre-identifying candidate grids containing potential $k$NNs, thus reducing unnecessary vertex visits. TEN*-Index achieves faster query response by leveraging precomputed $k$NN labels stored in a tree structure, avoiding graph traversal during query execution. 
As shown in the figures, D$k$NN exhibits lower query times than all baseline methods across the tested parameter ranges. By decomposing the query into parallel subgraph-level searches, D$k$NN effectively distributes the computational load which is not available in the centralized baseline algorithms. Moreover, as $k$ increases, the query time of D$k$NN grows at a slower rate, since a moderate rise in $k$ does not necessarily increase the number of subgraphs involved in the search. 
Fig.~\ref{FIG6}(d)–(g) further compare query efficiency under varying $\mu$, $\tau$, graph sizes, and number of queries. Across all settings, D$k$NN maintains a performance advantage, which indicates its robustness and scalability under diverse conditions.

\textbf{Update Efficiency Comparison.}
We compare the update efficiency of D$k$NN and baseline algorithms under varying edge update ratios $\alpha$, as shown in Fig.~\ref{FIG6}(h). This test directly measures each method's performance on the core challenge of dynamic road networks. While update times increase with $\alpha$ for all methods, D$k$NN achieves the shortest update time across all scenarios. It consistently outperforms the baselines by orders of magnitude. This significant performance gap validates D$k$NN's design for dynamic environments. By combining graph decomposition for parallel updates with a lightweight, cache-based adaptation, D$k$NN effectively minimizes the cost of handling network changes, demonstrating its efficiency for real-world dynamic applications.

\section{Conclusion}\label{Section6}
This work studies the $k$NN problem over moving objects on dynamic road networks where travel times fluctuate. Recognizing the maintenance challenges of index-based methods in dynamic road networks, we proposed D$k$NN, an algorithm that utilizes a distributed, parallel network expansion approach. D$k$NN partitions the network to enable parallel processing and then employs key mechanisms to overcome the challenges inherent to this distributed approach. It maintains global distance accuracy through message passing between border vertices during parallel subgraph explorations. Furthermore, it minimizes unnecessary work using a hybrid pruning strategy that combines Euclidean distance with the current $k$-th object's distance to prune irrelevant subgraphs. This design also facilitates provable query termination, which is managed by a token-matching mechanism. 
Implemented on the Storm platform, our experiments show that D$k$NN outperforms traditional methods in query efficiency and adaptability to dynamic changes. The results validate our distributed, index-free approach as an effective solution for $k$NN queries in real-world dynamic road network scenarios.

\section{Acknowledgement}
This work was supported by the Key Program of the National Natural Science Foundation of China  (62532002), and by the
National Natural Science Foundation of China (62172351).

\bibliographystyle{IEEEtran}
\bibliography{reference}

\end{document}